\documentclass[final]{dmtcs-episciences}

\usepackage[utf8]{inputenc}
\usepackage{subfigure}
\usepackage{amsmath}

\usepackage[round]{natbib}
\usepackage{tikz}

\tikzstyle{vertex}=[circle, draw, fill=black, inner sep=0pt, minimum size=7pt]
\usetikzlibrary{positioning,chains,fit,shapes,calc}
\newcommand{\vertex}{\node[vertex]}

\newtheorem{theorem}{Theorem}
\newtheorem{lemma}{Lemma}

\newcommand{\rtoldownimplies}{ \rotatebox[origin=c]{45}{$\Leftarrow$}}

\newcommand{\subseteqdown}{ \rotatebox[origin=c]{90}{$\supseteq$}}
\newcommand{\dotsdown}{ \rotatebox[origin=c]{90}{$\cdots$}}
\newcommand{\dotsdiag}{ \rotatebox[origin=c]{45}{$\cdots$}}

\author[G.Spencer]{Gwen Spencer\affiliationmark{1}}
\title[Sticky Seeding in
Discrete-Time 
Reversible-Threshold Networks]{Sticky Seeding in
Discrete-Time \\
Reversible-Threshold Networks}
\affiliation{Smith College, Northampton, MA, USA.}
\received{2015-7-3}
\accepted{2016-1-18}
\begin{document}
\publicationdetails{18}{2016}{3}{2}{1241}
\maketitle
\begin{abstract}
 When nodes can
repeatedly update their behavior (as in agent-based models from computational social science or repeated-game play settings) the problem of optimal network seeding becomes very complex. For a popular spreading-phenomena model of
binary-behavior updating based on thresholds of adoption among neighbors, we consider several planning problems in the design of \emph{Sticky Interventions}: when adoption decisions are reversible, the planner aims to find a Seed Set where temporary intervention leads to long-term behavior change. We prove that
completely converting a network at minimum cost is $\Omega(\ln (OPT) )$-hard to approximate and that maximizing conversion subject to a budget is $(1-\frac{1}{e})$-hard to approximate. Optimization heuristics which rely on many objective function evaluations may still be practical, particularly in relatively-sparse networks: we prove that the long-term impact of a Seed Set can be evaluated in $O(|E|^2)$ operations. For a more descriptive model variant in which some neighbors may be more influential than others, we show that under integer edge weights from $\{0,1,2,...,k\}$ objective function evaluation requires only $O(k|E|^2)$ operations. These operation bounds are based on improvements we give for bounds on time-steps-to-convergence under discrete-time reversible-threshold updates in networks.
\end{abstract}

\section{Introduction}
In the social and behavioral sciences there is a growing interest in the descriptive power of \emph{agent-based models}. The term \emph{agent-based} describes a broad class of models in which many independent agents repeatedly update their behavior in response to local interactions with other agents. After some time, emergent properties of the global system may be observed. Popular agent-based modeling tutorials by well-known social scientists and simulation researchers have been cited hundreds of times (such as \cite{RePEc:isu:genres:12515} and \cite{MacalNorth}).  Many features in agent-based models build on ideas that mathematicians and computer scientists have studied classically as \emph{cellular automata} or somewhat more recently as \emph{spreading-phenomena models}.

In fields where controlled experiments are very difficult
, agent-based models may offer a useful (if approximate) synthetic experiment
to shed light on the effects of proposed interventions.
Given an agent-based model,
a natural question follows: what does the model suggest about the form of the best interventions? What types of interventions will yield long-term improvements? Depending on the specification of agent behavior, and on how heterogeneously-structured the patterns of agent interaction are, this planning-oriented question can raise serious computational and algorithmic challenges. Most positive theoretical results we are aware of on maximizing influence in networks focus on a kind of ``one-shot" influence (adoption decisions are irreversible).
In this paper we focus on seeding interventions when nodes may repeatedly update their behavior (adoption decisions are reversible), as in most agent-based models.

In network settings, very few agent-based models proposed by social scientists yield optimization problems which have efficient accurate  algorithms\footnote{Or even satisfying approximation-algorithms which are efficient and provably near-optimal.}, so heuristic methods like genetic algorithms and local-improvement algorithms are a natural recourse. These heuristic methods may demand many function evaluations, and may require significant replication to ameliorate issues associated with settling at local optima while missing global optima. For example, running a local-improvement algorithm from 1,000 random initial solutions may give an experimenter substantial peace of mind. Unfortunately, for increasingly-complex agent-based simulations, even evaluating the predicted effect of an intervention may be quite slow/computationally costly (and not simply as a matter of poor implementation).  Some agent-based models may not even converge to a consistent ordering on the quality of interventions: how long should the agent behavior be allowed to evolve before the relative strength of two interventions is assessed?\footnote{This is particularly problematic as the choice of time-step interval is often somewhat artificial in agent-based models.} Optimizers and mathematicians interested in encouraging computationally-responsible reasoning in the social sciences (replication to reduce artifacts and the design of consistent evaluation metrics) should emphasize the variable computational costs associated with this diligence. Despite forbidding anecdotes, in many cases, replication and convergence to a consistent ordering are reasonable demands. In this paper, we prove hardness-of-approximation results for spreading a $\{0,1\}$-behavior with reversible adoptions, but we also give quadratic bounds on function evaluation and describe near-convergence from an arbitrary pattern of initial adoption.


Further, as
new sources of data from large online-social networks become available (e.g. variable strength of ties/variable sharing across links), what combinations of these these model features can be incorporated without jeopardizing manageable convergence (and compute) times? Researchers may preemptively limit descriptive power for fear that repeatedly evaluating ever-more complex simulations will be prohibitively slow for instances of real-world scale. Yet some powerful descriptive model features may carry only mild computational costs. For our model, adding a variable-integer-edge-weight feature causes only a modest increase in our function-evaluation bound. Including data-driven model features may increase the specificity of simulations describing real systems that social scientists care about.
We believe that there is a role for mathematical theory in helping scientists understand the computational costs of incorporating model features that capture new data sources. \\


\noindent\textbf{Our focus.}
Spreading-phenomena models of threshold-based binary-behavior in networks have a long history. Granovetter's foundational work in the sociology literature in the late 70s motivated computationally inquiry on the effect of variation in (often synthetic) network structure (see \cite{Granovetter:1978}). For example, Watts studied the distribution of long-term cascade sizes as a function of network structure under threshold-spread models (see \cite{export:164520}).
More recently, these models have been studied extensively in sociology as ``Complex Contagion." For example, see \cite{Complexcontagionsandtheweaknessoflongties}, and \cite{Cascadedynamicsofcomplexcontagion}. Threshold-based behavior updates also arise in behavioral economics models, as in \cite{Nyborg2006351}. Empirical work on Twitter data has suggested that the Complex-Contagion model is useful to understand spread of political activism in real-world contact networks (see \cite{Romero:2011:DMI:1963405.1963503}). Threshold decision rules also have provenance in game theory where they emerge as \emph{best-response strategies} for repeated game play in networks.\footnote{For background on game play in networks, see the widely-read textbook of \cite{David:2010:NCM:1805895}.}

We consider a model of binary $\{0,1\}$-behavior spreading in a network in which nodes update behavior based on a threshold of adoption among their neighbors. These threshold values are given explicitly in the input and may be heterogeneous. We assume that nodes will adopt a desirable behavior (Behavior 1) only when they observe enough of their neighbors choosing Behavior 1. If too few neighbors adopt Behavior 1, a node will adopt Behavior 0. As in the game-theoretic setting, we allow all agents to update concurrently and repeatedly.

If a planner can temporarily incentivize a small set of nodes to adopt Behavior 1, what does this threshold-based spread model suggest about form of highly-effective interventions? For a given budget there can exist a exponentially-large space of spatial interventions. Though the behavior-update process does not necessarily converge to a stable pattern, there is a somewhat satisfying notion of the long-term adoption caused by an intervention (this follows from a result of \cite{Dynamicsofpositiveautomatanetworks}).

We show that four natural variations of this intervention-planning problem are hard-to-approximate: no polynomial-time algorithms can exist without substantial compromise in the guaranteed quality of the solution. Can we still learn something about what these models predict about the form of effective interventions?


Given a behavior-spread model, nearly the bare minimum required to explore optimized interventions is the ability to repeatedly evaluate the objective function. For example, methods like local search algorithms and genetic algorithms involve repeated computation of the objective function value. Are heuristics which require many function evaluations reasonable? For discrete-time reversible-threshold binary-behavior updating we show  that evaluating the long-term effect of an intervention takes $O(|E|^2)$ operations where $E$ is the edge set for the network. The constant hidden by the big $O$ notation is relatively small (at most $8$ or $12$ for the most general cases), and for an edge-weighted variant (which can describe variable-strength ties in a social network) in which integer weights come from the set $\{1,2,3,...,k\}$ at worst function evaluations cost  $O(k|E|^2)$ operations, still with constant at most $8$ (or 12). The foundation of these bounds on compute time are new time-to-convergence results we prove for threshold automata in networks: behavior updating is guaranteed to converge in $2|E|+|V|$ time steps (or $2k|E|+|V|$ time steps for the edge-weighted case).
\vfill

\section{Model: Discrete-Time Reversible-Threshold Spread} \label{Modeldescrip}
\noindent \fbox{
\begin{minipage}{14cm}
\textbf{Input:} Finite network $G=(V,E)$. \\
Each node $i\in V$ has an integer threshold $b_i$ (the $b_i$ may be non-uniform). \\

\textbf{Initial State Vector:} At each time $t=\{0,1,2,3,...\}$, each $i\in V$ has one of two behaviors: Behavior 0 or Behavior 1. The \emph{Adoption Vector }, $x(t)\in \{0,1\}^{|V|}$, describes the pattern of adoption at time $t$. \footnote{If $x_i(t)=1$ this indicates node $i$ adopts behavior 1 at time $t$, etc.}  An initial adoption vector, $x(0)$, is given. Note that $\sum_{j\in V}x_{j}(0)$ may be non-zero.\\

\textbf{Planner Intervention at \emph{Seed Set:}} \\
The planner specifies a $d$-time-step intervention at subset  $V'\subseteq V$. \\
For all $i\in V'$, set $x_i(0)=1$.\\

\textbf{Evolution of Adoption:} \\
At each time $t\geq 1$, each $i\in V$ updates behavior:
\begin{itemize}
\item \textbf{Decision Rule:} If ($i\in V'$ AND $t\leq d-1$), then set $x_i(t)=1$.

\vspace{1mm}
Otherwise, node $i$ updates based on behavior of neighbors at $t-1:$\\
If $\sum_{j\in \delta(i)}x_{j}(t-1)\geq b_i$, then set $x_i(t)=1$.\\
If $\sum_{j\in \delta(i)}x_{j}(t-1)< b_i$, then set $x_i(t)=0$.
\end{itemize}
\end{minipage}}\\

\noindent These inputs and decision rule specify a sequence of binary vectors $x(0),x(1),x(2),...$ that describes the evolution of adoption of behavior in $G$.

Suppose that a planner can temporarily intervene to force a subset of nodes to adopt Behavior 1 for the $d$ time steps $\{0,1,2,...,d-1\}$ (at time $d$ these nodes resume normal updating). As other nodes update according to the decision rule, Behavior 1 may start to spread in the network. As is common, we refer to the subset of nodes where the planner intervenes as the \emph{Seed Set}. Subject to constraints on the size or cost of the Seed Set, our planner wants to choose the Seed Set which results in highest long-term adoption (long after the temporary intervention has ended). Qualitatively, the planner would like to choose a \emph{Sticky Intervention} for which $x(k)$ contains many 1s for arbitrarily large $k$. An intuitive definition of the \emph{Long-term Rate of Adoption} is
\[
\lim_{k\rightarrow \infty}\Big[\sum_{i\in V}x_i(k)\Big]
\]

Unfortunately, for the decision rule defined, $x(k)$ may not converge. When the limit above does not exist, $x(k)$ eventually alternates between two vectors (from \cite{Dynamicsofpositiveautomatanetworks})\footnote{For an example of a repeating cycle of length 2, see the appendix. In a computational setting, a useful stopping condition for function evaluation should incorporate this possible settling at a 2-cycle behavior.} , so define the \emph{Long-term Rate of Average Adoption} as
\[
\lim_{k\rightarrow \infty}\Big[\sum_{i\in V}\frac{1}{2}\Big(x_i(k)+x_i(k+1)\Big)\Big]
\]
Our planner seeks to maximize this well-defined objective function. There are two natural notions of what it means to \emph{intervene temporarily}.
\vspace{3mm}

\noindent\textbf{Duration of Intervention:}
\begin{itemize}
\item \textbf{temporary (t)}: The planner chooses a Seed Set and intervenes continuously until growth of Behavior 1 stops\footnote{Or, the growth of the 2-time-step average level of Behavior 1 stops.}, then stops intervening.
\item \textbf{fixed-duration(fd)}: The planner chooses a Seed Set and intervenes for $d$ consecutive time steps, then stops intervening.
\end{itemize}

\noindent \emph{Comment (Adoption is reversible):} In this model, nodes may update from Behavior $0$ to Behavior $1$, or from Behavior $1$ to Behavior $0$. In particular, updates from $1$ to $0$ may occur in cascades of Behavior 0 that start at $t=1$ based on an unstable pattern of adoption in the initial state vector, $x(0)$, or in cascades of Behavior 0 that start when the Seed Set $V'$ is no longer forced to Behavior 1.

\vspace{3mm}

\noindent\textbf{Planning Objective:}
\begin{itemize}
\item\textbf{Problem 1: Min-cost Complete Conversion (MCC)} What is the smallest cardinality Seed Set required to convert the entire network to Behavior 1 permanently?

\item \textbf{Problem 2: Budgeted Maximum Conversion (BMC)} If the Seed Set contains at most $k$ nodes, what is the maximum number of nodes that can be permanently converted to Behavior 1?
\end{itemize}

\section{Seeding the Stickiest Intervention is Hard to Approximate} \label{hardnesssec}

The following table summarizes our results on
hardness of approximation for the four variants of the seeding problem defined. \\

\noindent \textbf{Summary of Results on Hardness of Approximation:}\\
\noindent \begin{tabular}{|l|l|l|}\hline
& \textbf{temporary (t)} & \textbf{fixed-duration (fd)} \\ \hline
\textbf{Min-Cost} &    &   \\
\textbf{Complete Conversion}            &  $\Omega(\ln (OPT) )$)        & $\Omega(\ln (OPT) )$    \\
           &        &     (provided that $d\geq 2$)      \\ \hline
\textbf{Budgeted}&     &  \\
\textbf{Maximum Conversion}           &  $<(1-\frac{1}{e})$ $\approx 0.632$  &    $<(1-\frac{1}{e})$ $\approx 0.632$          \\
          &    &         (provided that $d\geq 2$)  \\ \hline
\end{tabular}
\vspace{4mm}

For Min-cost Complete Conversion, we prove that no polynomial-time algorithm guarantees a $O(\ln (OPT) )$-approximation (unless $NP$ has slightly superpolynomial time algorithms). This is by reduction from the Set Cover problem, for which Feige proved a $\ln (n) $ threshold for efficient approximation (see \cite{Feige:1998:TLN:285055.285059}).  Notably, the lower bound we give here is a function of the size of the optimal Seed Set.
For Budgeted Maximum Conversion, we prove that no polynomial-time algorithm can guarantee more than a $(1-\frac{1}{e})$-fraction of the optimal value. This is by reduction from the Maximum-Coverage Problem, where we again leverage a hardness result of Feige.

Before proving these hardness results, we comment on directions for positive algorithmic results.
Given our $(1-\frac{1}{e})$ inapproxamability result for Budgeted Maximum Conversion, readers may wonder whether the $(1-\frac{1}{e})$-approximation algorithm for the suggestively-named ``Linear-Threshold Model" of
\cite{Kempe05influentialnodes}
suggests a useful direction for the problems defined here. Unfortunately, the positive algorithmic result in Kempe, et al. relies on a very specialized model assumption about the form of uncertainty in threshold values. This assumption is used to prove that adoption is a submodular function of the set of seeds.\footnote{Kempe et al. assume that each node $i$ chooses a fractional threshold uniformly from $[0,1]$. Submodularity no longer holds if for arbitrary $\epsilon>0$,
each node $i$ chooses a fractional threshold uniformly from $[\epsilon,1]$. } Known results on maximizing submodular functions subject to budget constraints then immediately guarantee the success of certain greedy approaches to solution construction. In fact, submodularity-dependent analysis informs much of the theoretical computer science work on influence in networks (for example, also see \cite{Mossel:2010:SIS:1898187.1898191}). We emphasize that the objective in our problems is not submodular, and greedy approaches may fail dramatically for very simple examples.

Consider the \emph{Set Cover Problem} defined as follows. Let $S$ denote a set of elements $\{1,2,...,n\}$. Let $F$ denote a group of subsets of $S$ which we will denote $J_1,J_2,...,J_{|F|}$. Call
a set of indices $I$ a \emph{Set Cover} if
\[
\cup_{i\in I}J_i=S.\\
\]
The goal is to find a Set Cover of minimum cardinality.

\begin{theorem}
The Set Cover Problem can be reduced in polynomial time to an instance of Min-cost Complete Conversion (for any period of seeding $d\geq 2$). As a result, Min-cost Complete Conversion is $\Omega(\ln (OPT) )$ hard to approximate.
\end{theorem}

\begin{proof} Given an arbitrary instance of the Set Cover Problem, construct an instance of Min-cost Complete Conversion as follows. For each element $i\in S$, create a node $x_i$. Denote this set of ``element nodes" by $V_S$. For each $J_k\in F$ create a node $x_{J_k}$. Denote this set of ``subset nodes" by $V_F$. Our constructed instance has node set $V=V_S\cup V_F$. For every $(i, J_k)$-pair with $i\in S$ and $J_k\in F$, if $i\in J_k$ then include edge $(x_i, x_{J_k})$ in edge set $E$. Notice that $(V,E)$ gives a bipartite graph. Let the threshold for each node $x_i\in V_S$ be 1. Let the threshold for each node $x_{J_k}\in V_F$ be the degree of $x_{J_k}$, which by construction is $|J_k|$.

\begin{figure}[t!]
\fbox{
\begin{minipage}{14 cm}
\vspace{4mm}

\[\begin{tikzpicture}
	\vertex (x1) at (-3,2) [label=below:$x_i$] {};
	\vertex (x2) at (-2,2) [label=below:$x_a$] {};
	\vertex (x3) at (-1,2) [label=below:$x_b$] {};
	\vertex (x4) at (-0,2) [label=below:$x_c$] {};
	\vertex (x5) at (2,2) [label=below:$x_z$] {};
	\vertex (x6) at (-2,5) [label=above:$x_{J_k}$] {};
	\vertex (x7) at (-1,5) [label=above:$x_{J_l}$] {};
	\vertex (x8) at (0,5) [label=above:$x_{J_m}$] {};
	\vertex (x9) at (2,5) [label=above:$x_{J_q}$] {};
	\path
		(x1) edge (x6)
		(x2) edge (x6)
		(x4) edge (x6)
		(x4) edge (x7)
		(x3) edge (x8)
		(x2) edge (x7)
		(x1) edge (x8)
		(x5) edge (x8)
		(x9) edge (x3)
	 ;
	
	 \node[align=center, below] at (-4.5,2.2){\huge$V_S$};
	 \node[align=center, below] at (-3.5, 5.5){\huge $V_F$\normalsize};

	 \node[align=center, below] at (1,2){$...$};
	 \node[align=center, below] at (1,5.5){$...$};
	
	 \node[align=left, below] at (6,4.8){and $S=\{i,a,b,c,...,z\}$.};
	 \node[align=left, below] at (6, 5.5){Where $F=\{J_k,J_l,J_m,...,J_q\}$,};
	
	 \node[align=center, below] at (6,2){For example, $J_k=\{i,a,c\}$, etc.};

\end{tikzpicture}\]
\caption{\textbf{Schematic of constructed bipartite graph.} The constructed instance of Min-cost Complete Conversion. Each node in $V_S$ has threshold 1. Each node in $V_F$ has threshold equal to the cardinality of the corresponding subset $J_k$. }
\end{minipage}
}
\end{figure}
We make two observations about this class of constructed instances. First, for \emph{any} seed set in the constructed instance, all updates to Behavior 1 that will ever occur happen by the end of $t=2$. To see this, consider $V_S$ and $V_F$ separately. Suppose $x_i\in V_S$ is not a seed: $x_i$ can only be converted to Behavior 1 if there exists some $J_k\in F$ with $i\in J_k$ that has $x_{J_k}$ a seed. Otherwise, $x_i$ is adjacent only to non-seed subset nodes: due to the construction of the thresholds for $V_F$, such non-seed subset nodes are only converted to Behavior 1 when \emph{every node corresponding to a contained element already has Behavior 1}. Since $x_i$ does not have Behavior 1, such a subset node will not be converted to Behavior 1 (and consequently will never cause nodes corresponding to its elements to adopt Behavior 1). Thus, any non-seed $x_i\in V_S$ that adopts Behavior 1 will do so because it is adjacent to a seed node. Since the threshold for $x_i\in V_S$ is 1, adoption of Behavior 1 must happen at $t=1$. Since any subset node is only adjacent to element nodes, if the final non-trivial updates for element nodes occur at $t=1$, then the final non-trivial updates for a subset node must be at $t=2$.

Second, suppose that a seed set contains a subset node $x_{J_k}$. If $x_{J_k}$ is forced to adopt Behavior 1 for at least $t\in \{0,1\}$ then $x_{J_k}$ must adopt Behavior 1 for all $t$. This follows from our construction of the thresholds. Since $x_{J_k}$ has Behavior 1 at $t=0$, at $t=1$ all element nodes adjacent to $x_{J_k}$ will have Behavior 1. As a result, at $t=2$, $x_{J_k}$ will freely choose Behavior 1 without being forced, and also, all element nodes adjacent to $x_{J_k}$ will freely choose to adopt Behavior 1 (since $x_{J_k}$ was forced to Behavior 1 at $t=1$). At later time steps updates are trivial: at $t$, $x_{J_k}$ observed all neighbors adopting Behavior 1 at $t-1$, and all element nodes adjacent to $x_{J_k}$ observed $x_{J_k}$ adopting Behavior 1 at $t-1$. Effectively, seeds which are subset nodes quickly ``self-stabilize," and do the same for their element-node neighbors.

Now solve the constructed instance of Min-cost Complete Conversion where $d\geq2$ (the Seed Set will at least be forced to Behavior 1 for $t\in\{0,1\}$). Call the returned Seed Set $Q$. We explain how to massage $Q$ to find a seed set that only contains subset nodes which has size at most $|Q|$ and still converts all of $V$ to Behavior 1. This massaged $Q$ will have a natural interpretation as a Set Cover for $S$ of cardinality at most $|Q|$.

Suppose that the seed set $Q$ contains $x_i\in V_S$. If all neighbors of $x_i$ are also seeds, then $x_i$ can be removed from $Q$ to obtain a strictly smaller seed set that converts all nodes by $t=2$ (the only change is that $x_i$ will now convert to Behavior 1 at $t=1$ causing at most a 1 time step delay in other updates to Behavior 1). Otherwise there exists some $x_{J_l}\notin Q$ with $i\in J_l$. In this case, massage $Q$ by removing $x_i$ and adding such a $x_{J_l}$. Any subset node whose adoption of Behavior 1 depended on $x_i$ still adopts Behavior 1 by $t=2$ (since now $x_i$ adopts Behavior 1 at $t=1$, as $i\in J_l$). Any other node updates altered by this substitution result in Behavior 1 at a node in the place of Behavior 0. Thus, the massaged version of $Q$ still converts all of $V$ to Behavior 1 by $t=2$ (by construction of the node thresholds, full adoption is stable for $t\geq 3$).
Repeat this removal/massaging procedure until $Q$ contains only nodes corresponding to subsets from $F$.

Now interpret the subsets corresponding to nodes of massaged $Q$ as a proposed set cover $I$ of cardinality at most $|Q|$. As verified above, the massaged seed set $Q$ still converts all nodes of $V$ to Behavior 1 by $t=2$. Since non-seed element nodes in $V_S$ are converted only by adjacency to subset nodes which are seeds (from our first observation), every element node in $V_S$ must be adjacent to some seed from massaged $Q$. By construction of edge set $E$, this means that all $i\in S$ appear in some subset indexed by $I$. Thus, $I$ is a set cover containing at most $|Q|$ sets from $F$.

Finally, the optimal value for the Set Cover problem cannot be strictly less than the Min-Cost Complete Conversion optimal value, since every set cover $L$ with
$\cup_{i\in L}J_i=S$
corresponds to a seed set in our constructed instance of Min-Cost Complete Conversion which converts all of $V$ to Behavior 1 (all element nodes are converted to Behavior 1 at $t=1$ due to the covering property, and consequently all non-seed subset nodes are converted at $t=2$). 

Due to the correspondence demonstrated between solutions for an arbitrary Set Cover instance and solutions of the same numerical value for our polynomially-constructed instance of Min-cost Complete Conversion (for arbitrary $d\geq 2$), an $\alpha$-approximation algorithm for Min-cost Complete Conversion immediately gives an $\alpha$-approximation algorithm for Set Cover. Thus Min-cost Complete Conversion inherits $\Omega(\ln n )$-hardness from the Set Cover Problem where $n=|S|$ (this hardness holds unless $NP$ has slightly superpolynomial time algorithms, see \cite{Feige:1998:TLN:285055.285059} for details). Notice that $|S|$ corresponds to $|V_S|$ in our constructed Min-Cost Complete Conversion instance: our massaging procedure shows that $OPT\leq |V_S|$, so $\Omega(\ln (|V_S|))$-hardness of approximation certainly implies $\Omega\ln (OPT))$-hardness of approximation for Min-Cost Complete Conversion.
\end{proof}

A similar reduction (which introduces dummy nodes to stabilize element nodes and creates a complete-subgraph gadget on $V_F$ to avoid permanent conversion of subset nodes) shows that the Budgeted Maximum-Coverage Problem can be reduced in polynomial time to our Budgeted Maximum-Conversion Problem. Again, Budgeted Maximum Conversion inherits a hardness due to \cite{Feige:1998:TLN:285055.285059}: unless $P=NP$ there is no approximation algorithm that can guarantee a solution strictly better than a $(1-\frac{1}{e})$ fraction of optimal.

\begin{theorem}\label{bmcis1minus1overe}
The Budgeted Maximum Coverage Problem can be reduced in polynomial time to an instance of Budgeted Maximum Conversion (for any intervention length $d\geq 2$). As a result, Budgeted Maximum Conversion is $(1-\frac{1}{e})$-hard to approximate.
\end{theorem}

\noindent The details of this proof appear in the appendix.

\section{Computing the Effect of Intervention:\\
 Long-Term Average Adoption Rate}

Given our hardness results in the previous section, we turn our attention to the feasibility of optimization heuristics. Almost the bare minimum required by an optimization heuristic is the ability to repeatedly evaluate the objective function associated with a candidate feasible solution. In our case, what is the long-term effect of a particular seed set?  If function evaluations are too-expensive computationally, it may limit the size of instances where researchers can reasonably conduct computational studies. For Discrete-time Reversible-threshold binary behaviors we give an upperbound on the number of operations required to compute the Long-Term Average Adoption Rate which is quadratic in the cardinality of the edge set:

\begin{theorem}\label{computetime}
The long-term average adoption rate can be computed in $O(|E|^2)$ operations.
\end{theorem}

Theorem \ref{computetime} follows immediately from new bounds we give in Section \ref{proofofconvsec} on the number of time steps required before guaranteed convergence to the long-term average adoption rate. In each time step, the number of operations required to update all nodes according to adoption in their neighborhoods is $O(|E|)$, as each node $x$ must sum over $|\delta(x)|$ terms (for a total of $2|E|$ terms over all nodes).

The bound in Theorem \ref{computetime}  is for the more complicated cases: fixed-duration intervention, and temporary intervention from an arbitrary initial state vector $x(0)\in \{0,1\}^{|V|}$. For a temporary intervention maintained until growth of Behavior 1 stops starting from initial state with $||x(0)||=0$, convergence in $O(|V|)$ time steps to a stable pattern of long-term adoption is obvious: the evolution of adoption may be cleanly divided into a first phase is which adoption of Behavior 1 is growing (during the intervention), and a second phase in which adoption of Behavior 1 is being eroded (after the intervention). Each phase lasts at most $|V|$ time steps: if a time step elapses with only trivial behavior updates (each node maintains their behavior from the previous time step) then no non-trivial updates can possibly occur in future time steps.

In the the more complicated cases, adoption of Behavior 1 and erosion of Behavior 1 can occur simultaneously in $G$. Bounding time-until-convergence is significantly more subtle for these cases.\footnote{Note that when $||x(0)||=0$, the temporary variant corresponds to a special case of the fixed-duration variant when $d\geq n$.} Further, in these  cases the long-term rate of adoption may not be well defined. Instead, a simple long-term average rate may be reliably computed after convergence to a cycle that oscillates between 2 adoption vectors. We summarize our results in the table below.\\

\noindent \textbf{Summary of Results on Stability and Convergence to Long-term Average Adoption:}\\
\noindent \begin{tabular}{|l|l|l|}\hline
&\textbf{Discrete-Time Reversible-} & \textbf{Weighted-Neighbor }\\
&\textbf{Threshold Seeding (DRSeed)} & \textbf{Variant of DRseed} \\\hline
\hline
\textbf{temporary (t)} & converges to: stable adoption vector     & converges to: stable adoption vector  \\
 \emph{Given $||x(0)||=0$}          & convergence time bound: $2|V|$ & convergence time bound: $2|V|$            \\ \hline
\hline
\textbf{temporary (t)}&  converges to: 2-cycle  (from \cite{Dynamicsofpositiveautomatanetworks})   & converges to: 2-cycle  (from \cite{Dynamicsofpositiveautomatanetworks}) \\
  \emph{From arbitrary $x(0)$}        & convergence time bound: $2(2|E|+|V|)$ & convergence time bound: $2(2k|E|+|V|)$         \\

        &  & \emph{(for edge-weights from $\{0,1,2,...,k\}$) }       \\ \hline
\hline

\textbf{fixed-duration (fd)}&  converges to: 2-cycle  (from \cite{Dynamicsofpositiveautomatanetworks})   & converges to: 2-cycle  (from \cite{Dynamicsofpositiveautomatanetworks}) \\
  \emph{From arbitrary $x(0)$}         & convergence time bound: $d+2|E|+|V|$ & convergence time bound: $d+2k|E|+|V|$         \\

        &  & \emph{(for edge-weights from $\{0,1,2,...,k\}$) }       \\\hline
\end{tabular}

\vspace{3mm}

Before the proofs for the fixed-duration cases we make the following comment.  How much could our convergence bounds be improved? Though the adoption vector evolves in an exponentially-large space, for temporary interventions our bound on convergence time is linear in the number of nodes. In the more general fixed-duration setting we give an upperbound on convergence time that is linear in the number of edges in the network. Considering a network which is a simple line graph shows that the best possible upper bound on adoption-vector convergence time is $|V|$. We note that the average degree in social networks is often bounded by a constant\footnote{This is often understood as a consequence of Dunbar's social brain hypothesis (see \cite{DUNBAR1992469}).}: for such networks all upper bounds given are linear in the number of individuals in the network, so that at most a constant-factor improvement could be given.

In the next section, we harness special properties of our model to tighten a convergence argument given in \cite{Dynamicsofpositiveautomatanetworks} for a more general class of threshold automata. In particular,
our bound is linear in the size of the edge set of the network regardless of degree distribution (whereas \cite{Dynamicsofpositiveautomatanetworks} stated looser bounds and only claimed linear convergence time for uniform-degree networks).
Our proof extends almost immediately to the case in which edges have weights from the set $\{0,1,2,...,k\}$ for constant $k$.

\subsection{Convergence from an Arbitrary Pattern of Initial Adoption}\label{proofofconvsec}

For adoption vector $x(t)\in \{0,1\}^{|V|}$, let $N_{x(t)}$ denote the set of indices $i$ for which $x_i(t)=1$. The following lemma describes a basic property of our economic-threshold decision rule (which is not true of the more general class Goles studied). We establish it now for use near the end of our convergence-bound proof.

\begin{lemma} (Monotonicity of the future adoption set in the current adoption set)\label{monoton}
Let $x(t)$ and $x'(t)$ denote two adoption vectors at time $t$. Let $x(T)$ denote the result of applying the update rule to $x(t)$ for $T-t$ steps, and let  $x'(T)$ denote the result of applying the update rule to $x'(t)$ for $T-t$ steps.   If $N_{x(t)}\subseteq N_{x'(t)}$, then for all $T>t$, $N_{x(T)}\subseteq N_{x'(T)}$.
\end{lemma}

\begin{proof} \emph{(of Lemma \ref{monoton})} By induction, starting at time $t$.  From the assumptions of the theorem, $N_{x(t)}\subseteq N_{x'(t)}$.  Induction hypothesis: suppose that for time $T-1$ we have $N_{x(T-1)}\subseteq N_{x'(T-1)}$. Consider time step $T$. For each node $i\in V$:
if $i \in N_{x(T)}$, then it must be the case that at least $b_i$ neighbors of $i$ are in $N_{x(T-1)}$. Since all nodes which are in $N_{x(T-1)}$ are also in $N_{x'(T-1)}$ (by the induction hypothesis), at least $b_i$ neighbors of $i$ are in $N_{x'(T-1)}$. By the update rule for $i$, we get that $x_i'(T)=1$, so that $i\in N_{x'(T)}$. Thus, $N_{x(T)}\subseteq N_{x'(T)}$. 
\end{proof}

\begin{theorem} \label{converged}\textbf{(Main Theorem: Convergence Bound for Reversible Economic-Threshold Spread)} \\
In graph $G=(V,E)$, given an arbitrary initial adoption vector $x(0)\in \{0,1\}^{|V|}$:\\
within $2|E|+|V|$ time steps the evolving adoption vector will converge to a cycle of length at most 2.
\end{theorem}

Next, we explain the proof of this main theorem, which requires a number of intermediate lemmas.
The decision rule of our model
is a special case of the general threshold-based update rule analyzed in \cite{Dynamicsofpositiveautomatanetworks}. We follow Goles analysis closely, but by exploiting the monotonicity of our restricted case (established in Lemma \ref{monoton}), and emphasizing a more combinatorial
description of a key function, we give significantly tighter results for our model.

 Proceeding forward, we use the fact (from \cite{Dynamicsofpositiveautomatanetworks}) that for a specified set of integer thresholds denoted $b_i$ for $i \in V$  (our model as described until now), replacing $b_i$ with $b_i-0.5$ for all $i\in V$ gives an update procedure indistinguishable from updating that  uses the original $b_i$.\footnote{E.g. replacing a threshold of 3 at node $i$ with a threshold of 2.5 changes nothing about how the adoption vector will change.} Thus, without loss of generality, we assume all thresholds are half-integer.

The following lemma is directly from \cite{Dynamicsofpositiveautomatanetworks}, but we include a proof for completeness.

\begin{lemma} \label{enterscycle}(Adoption enters a cycle, \cite{Dynamicsofpositiveautomatanetworks})
Given any $x(0)\in \{0,1\}^{|V|}$,
there exists an integer  $c$ with the property that $x(t)=x(t+c)=x(t+2c)=...$ and that $x(t)$ is not equal to any of $x(t+1), x(t+2),...,x(t+c-1)$ for all t above some \emph{Transient Time } $T$.
\end{lemma}

\begin{proof} \emph{(of Lemma \ref{enterscycle})} The space of possible adoption vectors, $\{0,1\}^{|V|}$, is finite. Thus, there exists a time step at which some adoption vector $y\in\{0,1\}^{|V|}$ occurs for the second time. Since the update process is deterministic, the behavior from that point forward will be identical to the evolution after the first occurrence of $y$, giving a cycle.
\end{proof}

\noindent We'll define and analyze a special function $E(x(t))$. In contrast to the analysis in \cite{Dynamicsofpositiveautomatanetworks}, we describe $E(x(t))$ in terms of the update process in the network. To do this we introduce the idea of \emph{sightings}. Given an adoption vector $y \in \{0,1\}^{|V|}$, we say that node $i$ \emph{sights} each of its neighbors which is 1 according to y.\footnote{In terms of our update process:
are the number of sightings $i$ makes at $x(t)$ greater than
$b_i$? If so, then $i$ is 1 at time $t+1$.} The expression for $E(x(t))$ will include two terms related to sightings.\\

\noindent \textbf{Sightings Necessary to get $x(t)$: } \\
If $x(t)$ occurs during our updating process, we can give a bound on how many sightings must have happened
in the graph at time $t-1$: if $i$ is 1 at $t$, it must have sighted at least $b_i$ neighbors which were 1.
Thus, throughout the graph at least
\[
\sum_{i=1}^{|V|}b_ix_i(t) = \langle b,x(t) \rangle =  \mbox{(Sightings Necessary to get $x(t)$) }
\]
sightings must have happened at $t-1$. This counts accurately: $b_i$ is included exactly when $x_i(t)$ is 1.\\

\noindent \textbf{Sightings wasted in turning on $x(t+1)$:}\\
The adoption vector $x(t)$ produces the next state $x(t+1)$:
each node $i$ that is
1 in $x(t+1)$ saw at least $b_i$ sightings in $x(t)$, but $i$ may also have sighted
some extra neighbors at 1 beyond the $b_i$ that were required. We say \emph{these sightings were wasted in turning on $x(t+1)$}. Let $A_i$ denote the $i$th row of the adjacency matrix for $G$.\footnote{The adjacency matrix of $G$ is the $|V|\times|V|$ matrix that has $A_{ij}=1$ exactly when the edge $(i,j)$ is in $G$ and has all other entries 0.} The number of sightings that were wasted at node $i$ is
$
A_ix(t)-b_i$.

So, summing over $i$:
\begin{align}
\sum_{i=1}^{|V|} (A_ix(t)-b_i)(x_i(t+1)) &= \langle (Ax(t)-b), x(t+1) \rangle \\
&= \mbox{(Sightings wasted in turning on $x(t+1)$)}
\end{align}

\noindent This counts the correct quantity because $A_ix(t)-b_i$ is included in the sum exactly when $x_i(t+1)=1$. \\

\noindent \textbf{Defining $E(x(t))$:}
\begin{align}
E(x(t)):&= \mbox{(Sightings Necessary to get $x(t)$) } - \mbox{ (Sightings wasted in turning on $x(t+1)$)}\\
&= \langle b,x(t)\rangle-\langle (Ax(t)-b), x(t+1) \rangle
\end{align}

\begin{lemma} \label{decreasingunless2cycle} \cite{Dynamicsofpositiveautomatanetworks}
If the adoption vector has not entered a 2-cycle (aka, assuming $x(t)\neq x(t+2)$), then $E(x(t))$ is decreasing:
\begin{displaymath}
E(x(t+1))+0.5 \leq  E(x(t)).
\end{displaymath}
\end{lemma}

\begin{proof} \emph{(of Lemma \ref{decreasingunless2cycle})} The proof of this lemma is largely from \cite{Dynamicsofpositiveautomatanetworks}: we include the details for completeness and because they are critical to explaining the subsequent improvements we make.\\

\noindent We will show $E(x(t))-E(x(t+1))\geq 0.5$.
\begin{align}
E(x(t))&-E(x(t+1))=\\
&= \langle b,x(t)\rangle-\langle (Ax(t)-b), x(t+1) \rangle - \langle b,x(t+1)\rangle +\langle (Ax(t+1)-b), x(t+2) \rangle\\
&= \langle b,x(t)\rangle -\langle (Ax(t)), x(t+1) \rangle +\langle (Ax(t+1)-b), x(t+2) \rangle\\
&= \langle b,x(t)\rangle -\langle (Ax(t+1)), x(t) \rangle +\langle (Ax(t+1)-b), x(t+2) \rangle\\
&=-\langle (Ax(t+1)-b), x(t) \rangle +\langle (Ax(t+1)-b), x(t+2) \rangle\\
&=\langle x(t+2)-x(t), (Ax(t+1)-b)\rangle
\end{align}
In the algebra above: first we make a combined term for all sightings at $t$ by nodes which are 1 at $t+1$, next we switch the order of summation: sightings by 1s from $t+1$ at $x(t)$ are the same as sightings by 1s from $t$ at $x(t+1)$ (which is alternately true from $A$ a symmetric matrix), then simply combine terms.

\noindent The righthand side of the equality sums the following term over all nodes $i$:
\begin{align}
[x_i(t+2)-x_i(t)](A_ix(t+1)-b_i)
\end{align}
If $i$ has the same state in $x(t+2)$ and $x(t)$ then the term for $i$ has value 0. \\
If $x_i(t+2)=1$ and $x_i(t)=0$: $x_i(t+2)=1$  means $A_ix(t+1)-b_i\geq0.5$\\
If $x_i(t+2)=0$ and $x_i(t)=1$: $x_i(t+2)=0$  means $A_ix(t+1)-b_i\leq-0.5$\\

\noindent These facts follow from the half-integrality of the $b_i$. In both cases where $x_i(t+2)\neq x_i(t)$, the term for node $i$ contributes at least $1/2$ to the
value of $E(x(t))-E(x(t+1))$.  Since we assumed we are not in a 2-cycle yet (aka that $x(t)\neq x(t+2)$) there must be at least one $i$ that contributes value 1/2. This concludes the proof of Lemma \ref{decreasingunless2cycle}. \end{proof}

\noindent We can now prove that the evolving adoption vector will converge to a cycle of length at most 2.

\begin{proof}\emph{\textbf{(State of ``Convergence'' claimed in Theorem \ref{converged})}: }Suppose the cycle guaranteed by Lemma \ref{enterscycle} has length $c>2$ :
$x(t)=x(t+c)$ and $x(t)$ is not equal to any of $x(t+1), x(t+2),...,x(t+c-1)$. We can apply Lemma \ref{decreasingunless2cycle}:
\begin{align}
E(x(t))>E(x(t+1))>E(x(t+2))>...>E(x(t+c)).
\end{align}
Since $x(t)=x(t+c)$, we also have that $E(x(t))=E(x(t+c))$. This gives a contradiction. Thus $c\leq 2$.\end{proof}\\

\noindent It remains to prove how long this convergence will take. Following Goles, we showed in Lemma \ref{decreasingunless2cycle} that $E(x(t))$ decreases in every time step unless the process has converged to its final 2-cycle, so bounding the range of $E(x(t))$ will give an upper bound on the transient time of the process.  This is still our general strategy, but to give an improved upper bound over that from \cite{Dynamicsofpositiveautomatanetworks}, we use the monotonicity of our process (Lemma \ref{monoton}) to show that unless the process is already very close to a 2-cycle, the decrease in $E(x(t))$ is at least twice as large per time step as
specified by Lemma \ref{decreasingunless2cycle}:

\begin{lemma} \label{unlessclose}
Unless the evolving adoption vector is within $2|V|$ time-steps of entering a 2-cycle, at least 2 nodes $i$ have $x_i(t)\neq x_i(t+2)$, so that
\[
E(x(t+1))+1\leq  E(x(t)).
\]
\end{lemma}

\begin{proof} \emph{(of Lemma \ref{unlessclose})} Suppose that $x(t+2)$ differs from $x(t)$ in only one position.  We'll show that the adoption vector enters a 2-cycle within $2|V|$ time steps. Let $N_{x(t)}$ denote the subset of $V$ corresponding to nodes that are at 1 in $x(t)$. Suppose that $x_i(t+2)=1$ and $x_i(t)=0$, and for all other $j$, $x_j(t+2)=x_j(t)$. By definition, $N_{x(t)}\subseteq N_{x(t+2)}$. We use monotonicity (Lemma \ref{monoton}) to reason about $N_{x(\cdot)}$ in the following time steps.

Consider $N_{x(t+1)}$: it results from sightings at $N_{x(t)}$.  Since $N_{x(t)}\subseteq N_{x(t+2)}$, all sightings at $N_{x(t)}$ happen at $N_{x(t+2)}$.  Thus, $N_{x(t+3)}$ is a superset of $N_{x(t+1)}$: any node that decided to adopt based on $x(t)$ will certainly adopt based on $x(t+2)$. By the same rationale, $N_{x(t+1)}\subseteq N_{x(t+3)}$ gives that $N_{x(t+2)}\subseteq N_{x(t+4)}$. Then $N_{x(t+2)}\subseteq N_{x(t+4)}$ gives that $N_{x(t+3)}\subseteq N_{x(t+5)}$, etc.  This argument holds iteratively. To make this more clear, write the set of adopters in a z-pattern using $N_{x(t)}\Rightarrow N_{x(t+1)}$ to denote that the set of adopters at time $t+1$ results from the set of adopters at time $t$, and add the subset relationships:

\vspace{2mm}
\begin{center}
\begin{tabular}{lll}
\vspace{3mm}
$N_{x(t)}$              & $\Rightarrow$ &  $N_{x(t+1)}$     \\
\vspace{3mm}
\hspace{2mm} $\subseteqdown$         &$\rtoldownimplies$  & $\subseteqdown$ \\
\vspace{3mm}
$N_{x(t+2)}$            & $\Rightarrow$ &   $N_{x(t+3)}$   \\
\vspace{3mm}
\hspace{2mm} $\subseteqdown$         &$\rtoldownimplies$  & $\subseteqdown$\\
\vspace{3mm}
$N_{x(t+4)}$            & $\Rightarrow$ &  $N_{x(t+5)}$    \\
\vspace{3mm}
\hspace{2mm} $\subseteqdown$         &$\rtoldownimplies$ & $\subseteqdown$ \\
\vspace{3mm}
$N_{x(t+6)}$           & $\Rightarrow$ &  $N_{x(t+7)}$ \\
\vspace{3mm}
\hspace{7mm}$\dotsdown$          &$\dotsdiag$ & \hspace{5mm}$\dotsdown$
\vspace{1mm}
\end{tabular}
\end{center}
Since the update process is deterministic, if any of these $\subseteq$ relationships is not proper then the adoption vector has entered a 2-cycle. That is, if $N_{x(t+k)}=N_{x(t+k+2)}$, then it also must be the case that $N_{x(t+k+1)}=N_{x(t+k+3)}$, so that for all time steps after $t+k$ the adoption vector alternates between $x(t+k)$ and $x(t+k+1)$. The maximum number of time steps that could elapse before a $\subseteq$ relationship is forced to be non-proper is $2|V|$: since all the subsets must be proper, at least one node is in $(N_{x(t+k+2)}\setminus N_{x(t+k)})$ for all $k$. The longest path that could ever be achieved is if $N_{x(t)}$ is the empty set and $N_{x(t+2|V|)}$ is the entire set $V$.  After this many time steps it must be the case that $N_{x(t+k+2)}=N_{x(t+k)}$: the adoption vector has entered a 2-cycle within $2|V|$ time steps.

A symmetric argument (with opposite direction of $\subseteq$ relationships) establishes the case where $x_i(t+2)=0$ and $x_i(t)=1$, and for all other $j$, $x_j(t+2)=x_j(t)$.
We have established that if $x(t+2)$ differs from $x(t)$ in only one position the adoption vector will enter a 2-cycle within $2|V|$ time steps. Equivalently, if the adoption vector is more than $2|V|$ time steps from entering a 2-cycle, then it must be that $x(t+2)$ differs from $x(t)$ in strictly more than 1 position.  Thus, evaluating the final expression from the proof of Lemma \ref{decreasingunless2cycle}:
\begin{align}
E(x(t))&-E(x(t+1))=\langle x(t+2)-x(t), (Ax(t+1)-b)\rangle \geq 1.  
\end{align}
This concludes the proof of Lemma \ref{unlessclose}.\end{proof}

\begin{lemma}\label{range}
The range of values that $E(x(t))$ can achieve is $\leq 2|E|$.
\end{lemma}

\begin{proof} \emph{(of Lemma 5)}
\noindent Recall the definition of $E(x(t))$:
\begin{align}
E(x(t))=
(\mbox{Sightings Necessary to get } x(t))-(\mbox{Sightings wasted in turning on } x(t+1))
\end{align}

\noindent From our definitions of necessary and wasted sightings, the first term is always positive (or 0) and the second term is always negative (or 0). Letting the first term be as large as possible, and the second term be as small as possible, we obtain an upper bound on $E(x(t))$ of $\sum_{i=1}^{|V|} b_i$ (every node in $|V|$ makes the necessary sightings for it to adopt).

An obvious lower bound on $E(x(t))$ assumes the first term is $0$ and makes the second term as large in magnitude
as possible (as many wasted sightings as possible at every node):
\begin{align}
E(x(t))\geq -\sum_{i=1}^{|V|} (\mbox{deg}(i)-b_j) = -\sum_{i=1}^{|V|} (\mbox{deg}(i))+\sum_{i=1}^{|V|} b_i.
\end{align}
Thus the range of $E(x(t))$ is at most the upper bound minus the lower bound:
\begin{align}
\sum_{i=1}^{|V|} b_i- \Big( -\sum_{i=1}^{|V|} (\mbox{deg}(i))+\sum_{i=1}^{|V|} b_i\Big)=\sum_{i=1}^{|V|} (\mbox{deg}(i))=2|E|. 
\end{align}
\noindent This concludes the proof of Lemma \ref{range}.\end{proof}

\begin{proof} \emph{\textbf{(Time to Convergence claimed in Theorem \ref{converged})}}
The range of $E(x(t))$ is at most $2|E|$ from Lemma \ref{range}.  At most $2|V|$ time steps can decrease $E(x(t))$ by only 1/2.  Every other time step must result in a decrease in $E(\cdot)$ of size at least 1.  Thus, from any initial adoption vector $x(0)$ there are at most $(2|V|+ (2|E|-|V|)/1)=2|E|+|V|$ time steps before the adoption vector enters a stable state or a 2-cycle. 
\end{proof}\\

\noindent\textbf{Extension: When Some Relationships Are More Influential.}
To our model input, add that each edge $e\in E$ has a weight $w_e$ from the set of integers $\{0,1,2,...,k\}$, and modify the decision rule as follows.

\begin{itemize}
\item \textbf{Decision Rule:} For each node $i\in V$, if the edge-weighted sum of neighbors of $i$ who adopt Behavior 1 at $t-1$ is at least $b_i$, then node $i$ adopts Behavior 1 at time $t$. Otherwise node $i$ adopts Behavior 0 at time $t$.
\end{itemize}

\noindent Our proof generalizes immediately. The matrix $A$ is now the weighted adjacency matrix. The definition of $E(x(t))$ is generalized so the terms describe the weighted amount of sightings to get $x(t)$ and the wasted amount of weighted sightings to get $x(t+1)$.  Sums of degrees become sums of weighted degrees. The range of $E(x(t))$ is now $2\sum_{e\in E}w_e\leq 2k|E|$.\\

\begin{theorem} \textbf{(Convergence Bound: Weighted-Neighbor Reversible-Threshold Updating)}.\\
In graph $G=(V,E)$, where each edge has weight $w_e$ from $\{0,1,2,...,k\}$, given an arbitrary initial adoption vector $x(0)\in \{0,1\}^{|V|}$:
within $(2 \sum_{e\in E}w_e+ |V|)\leq (2k|E|+|V|)$ time steps the evolving adoption vector will converge to a cycle of length at most 2.
\end{theorem}

\section{Conclusion and Future Directions}

Motivated by the prevalence of repeated behavior updating in the computational social sciences literature, we considered the planning problem of designing \emph{Sticky Seeding Interventions in Networks} under reversible-threshold discrete-time updating. We proved that several natural variants of the planning problem are hard-to-approximate (by reductions from Set Cover), but we also provided quadratic bounds on function evaluation (even from arbitrary initial states of adoption). Thus, optimization heuristics that repeatedly evaluate the objective may be practical, even when empirically-motivated features like heterogeneous edge weights are added to the model. Incorporating model features that represent new data sources (like variable edge-weights)
increases the specificity of predictions, potentially
allowing valuable contrasts in qualitative properties of optimized interventions across different networks.

Our long-term spread objective that considers reversible adoption and heterogeneous edge weights is already complicated to evaluate, even though the model we consider in this paper is entirely deterministic. A prominent alternative to the game-theory-style concurrent-updating we've assumed is \emph{random asynchronous updating} where nodes update in a random order one at a time.\footnote{Sometimes these updates follow a random permutation, though sometimes there is no control to maintain uniform frequency of updates across the node set.} It is easy to construct small examples where random asynchronous updating leads to strongly-different long-term behavior than concurrent updating (even in an expected-value sense), or where the randomly-realized long-term outcomes of asynchronous updating vary enormously.
It is unclear whether these types of differences might somehow be damped at a larger scale.
Efficient exact expected long-term spread evaluation in this asynchronous context seems impossible due to the exponentially-large space of possible update orders,  though perhaps a strong assumption about the structure of the input network (as frequently appears in the statistical physics literature) might provide some traction.

We close with a general comment. The study of spread in networks  gives rise to a number of fascinating theoretical and applied questions. Variations in model assumptions and network structure can cause dramatic changes in the qualitative behavior of the system and in the form of optimized interventions. Many positive theoretical results rely on rather-special assumptions (for example, very-specific structural restrictions, or highly-special forms of uncertainty in the input instance). A significant future challenge will be to understand which messages about networks are ``stable'' against some variation in these assumptions. While it may seem that theoretical reasoning, synthetic computational exploration, and scientific investigation of real networks are diverging fields, we believe that each of these areas has rich insights to offer to its counterparts.

\acknowledgements
\label{sec:ack}
The author is thankful the reviewers and editor for useful suggestions in the preparation of this manuscript.

\nocite{*}
\bibliographystyle{abbrvnat}
\bibliography{reviseddmtcsmasterbib}
\label{sec:biblio}

\vspace{5mm}

\section*{Appendix}
\subsection*{Example of 2-cycle}
As mentioned in Section \ref{Modeldescrip}, the adoption vector may not converge to a single repeating vector.

\begin{figure}[h!]
\fbox{
\begin{minipage}{14.5 cm}
\vspace{4mm}

\[\begin{tikzpicture}
	\vertex (vleft) at (-1.5,0) [label=left:$1$] {};
	\vertex (vmiddle) at (0,0) [label=below left:$2$] {};
	\vertex (vabove) at (0,1.5) [label=left:$1$] {};
	\vertex (vbelow) at (0,-1.5) [label=left:$1$] {};
	\vertex (vright) at (1.5,0) [label=right:$1$] {};
	\path
		(vleft) edge (vmiddle)
		(vabove) edge (vmiddle)
		(vright) edge (vmiddle)
		(vbelow) edge (vmiddle)
	 ;
\end{tikzpicture}\]
\caption{\textbf{(An alternating cycle of length 2 with average adoption $50\%$)} Nodes are marked with their thresholds. Start from an initial pattern of adoption in which only the center node has Behavior 1. The adoption pattern cycles between two vectors; average adoption is $50\%$.}
\end{minipage}}
\end{figure}
Figure 2 demonstrates an example in which iterated application of the decision rule does not converge to a single repeating adoption vector. In the example in Figure 2, an average rate of long term adoption can be computed because the adoption vector enters a small cycle (of length 2).

The long-term behavior of the update rule is always at least as stable as the example in Figure 2: after a sufficient number of time steps the adoption vector will either be a stable repeating vector or alternate between 2 adoption vectors.\footnote{This was true even for the more general class of threshold automata Goles studied.} We show in the main text that convergence to such an alternating state happens within $O(|E|)$-time steps (where the hidden constant is small).

\subsection*{Hardness of Budgeted Maximum Conversion}

Here we give the details of the proof of Theorem \ref{bmcis1minus1overe} from Section \ref{hardnesssec}.

Consider the \emph{Budgeted Maximum Coverage Problem} defined as follows. Let $S$ denote a set of elements $\{1,2,...,n\}$. Let $F$ denote a group of subsets of $S$ which we will denote $J_1,J_2,...,J_{|F|}$. Given a budget $k$, the objective is to specify a set of $k$ indices $I$ so that the cardinality of the following union is maximized:
\begin{align}
\cup_{i\in I}J_i.
\end{align}

\noindent \textbf{Theorem 2}
\emph{The Budgeted Maximum Coverage Problem can be reduced in polynomial time to an instance of Budgeted Maximum Conversion (for any intervention lasting at least 2 time steps). As a result, Budgeted Maximum Conversion is $(1-\frac{1}{e})$-hard to approximate.}

\begin{proof} \emph{(of Theorem \ref{bmcis1minus1overe})} Given an arbitrary instance of the Budgeted Maximum Coverage Problem, construct an instance of Budgeted Maximum Conversion as follows. For each element $i\in S$, create a node $x_i$. Denote this set of ``element nodes" by $V_S$. For each node $x_i \in V_S$ create a unique dummy node $y_i$. Denote this set of dummy nodes by $V_D$. For each $J_k\in F$ create a node $x_{J_k}$. Denote this set of ``subset nodes" by $V_F$.
Our constructed instance has node set $V=V_S\cup V_D\cup V_F$.
For every $(i, J_k)$-pair with $i \in S$ and $J_k\in F$: if $i\in J_k$ then include edge $(x_i, x_{J_k})$ in edge set $E$.
For each $i\in S$, include the edge $(x_i,y_i)$ in $E$.
Finally, for each pair of subset nodes in $V_F$, include the edge between them in $E$.
Let the threshold for each node $x_i\in V_S$ be 1. Let the threshold for each node  $y_i\in V_D$ be 1. Let the threshold for each node $x_{J_k}\in V_F$ be the degree of $x_{J_k}$, which by construction is $|J_k|+|F|-1$. Let the budget for seeding be $k$ nodes (assume $k<|F|-1$ as otherwise the best Maximum-Coverage index set could be found by enumeration in polynomial time).

We make two observations about this class of constructed instances. First, no subset node will adopt Behavior 1 except in time steps in which it is forced to do so. This follows from our construction of the high thresholds for $x_{J_k}$: the maximum number of neighbors of $x_{J_k}$ adopting Behavior 1 is the number of seeds ($k$) plus the number of elements in $J_k$, for a total of $|J_k|+k<|J_k|+|F|-1$ (this is insufficient for $x_{J_k}$ to freely choose Behavior 1). Thus, even if all element nodes adopt Behavior 1, among the subset nodes only Seed nodes will adopt Behavior 1 (and only while the intervention is in place).

Second, suppose that a Seed Set contains a subset node $x_{J_k}$. If $x_{J_k}$ is forced to adopt Behavior 1 for at least $t\in \{0,1\}$ then all of the element nodes adjacent to $x_{J_k}$ must adopt Behavior 1 for all $t$. This follows from our construction of the thresholds. Since $x_{J_k}$ has Behavior 1 at $t=0$, at $t=1$ all element nodes adjacent to $x_{J_k}$ will choose Behavior 1. At $t=2$ each such element node observes that $x_{J_k}$ had (forced) Behavior 1 at $t=1$ and so continues to adopt Behavior 1. Also at $t=2$, the dummy node connected to each such element node will choose to adopt Behavior 1. Then, at $t=3$, regardless of adoption by $x_{J_k}$, each element node and its dummy will observe each other adopting Behavior 1 at $t=2$ and continue to adopt Behavior 1. All future time steps have the same trivial updates.
Effectively, seeds which are subset nodes quickly cause their corresponding (element node, dummy node)-pairs to ``stabilize" at Behavior 1.

Now solve the constructed instance of Budgeted Maximum Conversion where $d\geq2$ (the Seed Set will at least be forced to Behavior 1 for $t\in\{0,1\}$). Call the returned Seed Set $Q$. From the budget constraint for seeding we have $|Q|\leq k$. We will massage $Q$ to find a seed set that only contains subset nodes which has size at most $|Q|$ and still converts as many nodes to  Behavior 1 as $Q$ does. This massaged $Q$ will have a natural interpretation as a feasible solution for the Budgeted Maximum Coverage Problem for $S$.

For $i\in S$, suppose that either an element node $x_i$ or a dummy node $y_i$ is in $Q$. If all $J_k\in F$ that have $i\in J_k$ are in $Q$, then $x_i$ (or $y_i$ respectively) can be removed from $Q$ to obtain a strictly smaller Seed Set that converts to Behavior 1 all nodes that $Q$ does. Otherwise, there exists some $J_l\in F$ with $i\in J_l$ and $x_{J_l}\notin Q$. In this case, massage $Q$ by removing $x_i$ (or $y_i$ respectively) and adding subset node $x_{J_l}$.  The resulting set of seeds still converts both $x_i$ and $y_i$ to Behavior 1 (since $i\in J_l$ this happens by $t=2$ even when $d$ is as low as 2 by our second observation), and all other differences in node updates substitute Behavior 1 in the place of Behavior 0 (as now all element nodes-and their dummy copies- adjacent to $x_{J_l}$ will permanently adopt Behavior 1 from our second observation). Repeat this removal/massaging procedure until $Q$ contains only nodes corresponding to subsets from $F$.

Now interpret the subsets corresponding to nodes of massaged $Q$ as a proposed budgeted cover, $I$, of cardinality at most $|Q|\leq k$. As verified above, the massaged seed set $Q$ still converts all nodes to Behavior 1 that $Q$ did (and does so by $t=2$). The number of elements from $S$ covered by $I$ is precisely half the number of nodes converted in our constructed Budgeted Maximum Conversion instance.

Finally, the optimal value for Budgeted Maximum Coverage problem cannot be strictly more than half of the Budgeted  Maximum Conversion optimal value, since every budgeted cover $L$ whose union has
$|\cup_{i\in L}J_i|$ elements
corresponds to a seed set in our constructed instance of Min-Cost Complete Conversion which converts exactly $2|\cup_{i\in L}J_i|$ nodes to Behavior 1 (all element nodes corresponding to elements in $\cup_{i\in L}J_i$ are converted to Behavior 1 at $t=1$, and their dummy partners are converted to Behavior 1 at $t=2$, and these nodes adopt Behavior 1 in all future time steps from our second observation). 

Due to the correspondence demonstrated between solutions for an arbitrary Budgeted Maximum Coverage instance and solutions of exactly twice the numerical value for our polynomially-constructed instance of Maximum Budgeted Conversion (for arbitrary $d\geq 2$), an $\alpha$-approximation algorithm for Maximum Budgeted Conversion immediately gives an $\alpha$-approximation algorithm for Budgeted Maximum Coverage. Thus budgeted maximum conversion (for arbitrary $d\geq 2$) inherits $(1-1/e)$-hardness from the Budgeted Maximum Coverage Problem (this hardness holds unless $P=NP$, see \cite{Feige:1998:TLN:285055.285059} for details).
\end{proof}

A simpler reduction without the complete subgraph on $V_F$ is possible if the \emph{threshold for a subset node can exceed the degree of the node}. This seems somewhat abusive of the model, however, so we have included the argument as given.

\end{document}